\documentclass[journal]{IEEEtran}
\usepackage{authblk}
\usepackage[compact]{titlesec}
\titlespacing{\section}{1pt}{*1}{*1}
\titlespacing{\subsection}{1pt}{*1}{*0}
\titlespacing{\subsubsection}{1pt}{*0}{*0}
\usepackage[utf8]{inputenc}
\usepackage{setspace}
\usepackage{graphicx}
\usepackage{mathtools}
\usepackage{comment}
\usepackage{algorithm,algpseudocode}
\usepackage{subfigure}
\usepackage{graphicx}
\usepackage{float}
\usepackage{soul}
\usepackage{hyperref}
\usepackage{amsfonts}
\usepackage{booktabs}
\usepackage{amsmath}
\usepackage{bm}
\usepackage{mathtools}
\usepackage{tabularx}
\usepackage{array}
\usepackage{breqn}
\usepackage{supertabular}
\usepackage{setspace}
\usepackage{amsmath}
\usepackage{amsthm}
\usepackage[margin=0.5in]{geometry}
\newcolumntype{H}{>{\setbox0=\hbox\bgroup}c<{\egroup}@{}}
\usepackage{url}
\usepackage{xcolor}
\usepackage{makecell}
\usepackage{ragged2e}  
\newcolumntype{Y}{>{\RaggedRight\arraybackslash}X} 
\newtheorem{theorem}{Theorem}[section]
\usepackage[export]{adjustbox}
\newcommand{\R}{\mathbb{R}}

\usepackage{cite}
\usepackage{enumitem}

\usepackage[
singlelinecheck=false 
]{caption}
\usepackage[inline]{trackchanges}
\addeditor{NG}
\addeditor{PR}

\begin{document}
\title{Blockchain Based Decentralized Replay Attack Detection for Large Scale Power Systems}
\author{\IEEEauthorblockN{Paritosh Ramanan\IEEEauthorrefmark{1}\IEEEauthorrefmark{2},Dan Li\IEEEauthorrefmark{2} and Nagi Gebraeel\IEEEauthorrefmark{3}}
\thanks{This work was supported in part by the Georgia Power Faculty Fellowship; in part by the U.S. Department of Energy, Office of Science, Office of Advanced Scientific Computing Research, Applied Mathematics Program under Award DE-SC- 0016564; and in part by the Sam Nunn Security Program (SNSP) Fellowship.}
\thanks{\IEEEauthorrefmark{1}School of Computational Science and Engineering, Georgia Institute of Technology, Atlanta, GA, USA 30332}
\thanks{\IEEEauthorrefmark{3} Department of Industrial Engineering, College of Engineering, Computing, and Applied Sciences, Clemson University, Clemson, SC 29634}
\thanks{\IEEEauthorrefmark{3}H. Milton Stewart School of Industrial and Systems Engineering, Georgia Institute of Technology, Atlanta, GA, USA 30332.\\paritoshpr@gatech.edu,danli@gatech.edu,nagi@isye.gatech.edu}
}
\markboth{Accepted to IEEE Transactions on Systems, Man and Cybernetics: Systems}{Ramanan \MakeLowercase{\textit{et al.}}: Blockchain Based Replay Attack Detection}
\maketitle
\begin{abstract}
Large scale power systems are comprised of regional utilities with  assets that stream sensor readings in real time. In order to detect cyberattacks, the globally acquired, real time sensor data needs to be analyzed in a centralized fashion. However, owing to operational constraints, such a centralized sharing mechanism turns out to be a major obstacle. In this paper, we propose a blockchain based decentralized framework for detecting coordinated replay attacks with full privacy of sensor data. We develop a Bayesian inference mechanism employing locally reported attack probabilities that is tailor made for a blockchain framework. We compare our framework to a traditional decentralized algorithm based on the broadcast gossip framework both theoretically as well as empirically. With the help of experiments on a private Ethereum blockchain, we show that our approach achieves good detection quality and significantly outperforms gossip driven approaches in terms of accuracy, timeliness and scalability.
\end{abstract}
\begin{IEEEkeywords}
blockchain, decentralized analytics, data privacy, power systems, global replay attacks
\end{IEEEkeywords}
\section{Introduction}
Affordable sensor and communication technologies have given rise to a growing wave of industrial digitization. The power industry has been at the forefront of this trend that has culminated into a digital transformation of the power grid. Such levels of digitization have led to automation and digital control components that are collectively referred to as Industrial Control Systems (ICS). Until recently, ICS used specialized communication and control protocols that made them relatively immune to cyberattacks. However, with increase in Industrial Internet-Of-Things (IIoT) enabled assets, traditional ICS have gradually become heavily integrated with standard IT components. Meanwhile, on the physical level, the grid has evolved into a complex network with a high degree of interdependency among utility providers \cite{rosas2007topological, kinney2005modeling}. Such an interdependent, digitized grid has an increased vulnerability to various kinds of cyberattacks \cite{krutz2005securing}. Tackling these vulnerabilities requires competing utilities to share sensitive information with a trusted, centralized entity that can quickly assess cybersecurity related threats. The process of sharing data with a centralized entity can be challenging for many utilities due to the presence of a single point of failure, privacy concerns and competing market dynamics. Therefore, this paper proposes a blockchain based, decentralized methodology for detecting the probability of a global network cyberattack while preserving data privacy. Our approach is particularly useful in the context of IIoT enabled assets that yield real time sensor data to help monitor the power network.

A significant portion of ICS-focused cyberattacks involve data manipulation. Often the intent of these attacks is to impact asset reliability either by accelerating physical or efficiency degradation or causing sudden breakdowns. The \textit{Stuxnet} worm has often been referenced as a classic example of an ICS-focused cyberattack with data manipulation \cite{kushner2013real, albright2010did}.
Another popular example was the Aurora Generator Test in which a 2.25 MW substation generator was destroyed through a planned cyberattack that caused an out-of-sync closing of protective relays \cite{weiss2016aurora, rid2012cyber}

ICS cyberattacks involving data manipulation have been classified into three major types, false data injections \cite{liu2011false}, replay attacks \cite{peng2019survey},\cite{mo2014detecting} and covert attacks \cite{smith2011decoupled}. This paper considers ICS replay attacks where a malicious agent replays sensor measurements representing normal operating conditions in order to mask underlying malicious control actions. We focus on investigating coordinated, large scale scenarios where an ICS replay attack is mounted on more than one regional utility provider in the power network. We refer to such attacks as global network cyberattacks, or global attacks for short. ICS replay attacks are difficult to detect. Most of the existing detection algorithms do not distinguish between replay cyberattacks and naturally occurring equipment or controller faults. 

{
Current state-of-the-art ICS attack detection methods are intended for individual plant sites and assets only. Utility stakeholders typically rely solely on such localized detection mechanisms to raise alarms. However, without the knowledge of a rapidly evolving state of global cyber health, the local mechanism will only raise an alarm when local assets are either i) under an attack themselves, or ii) are facing the system repercussions of a network-wide attack. By the time a local alarm has been raised, utilities would have lost precious time which could otherwise be used to localize, diagnose and isolate faults. Moreover, local detection schemes are inherently based on some underlying hypothesis test, which means that there is always some level of unavoidable false alarms. In large scale settings like the one considered in this paper, detection schemes could significantly overstate the cyber threat level since false alarm rates increase with the number of hypotheses tests being conducted \cite{benjamini2001control}. Without active collaboration with other stakeholders, utilities might waste resources trying to diagnose false alarms unnecessarily. Therefore, a major stumbling block for global attack detection arises due to the inability of different utility stakeholders to collaborate and share their local ICS attack detection insights and alarms. Such roadblocks are inevitable despite the fact that stakeholders might be equipped with state of the art ICS attack detection mechanisms. 

Consequently, utility stakeholders are left with a tough choice between one of two options. Either, they pursue a purely localized approach with state of the art ICS attack detection, not withstanding its various blinding limitations. Alternatively, they could adopt a real time sensor and cyber incident data sharing strategy that relies on a centralized coordinator. The centralized entity would be in charge of real time processing and analytics of sensor data being streamed from a large number of assets, while providing alerts pertaining to the global cyber health of the network as a whole. Many such centralized data sharing programs have been pursued in the past. For instance, the U.S. Department of Energy operates the Cyber Risk Information Sharing Program (CRISP) \cite{multiYearPlan} which provides utility members with a platform to share cyber incident/sensor data. 

However, such centralized data sharing programs come with their own set of challenges. First, sharing data through a centralized repository presents several privacy risks along with efficiency and agility issues \cite{fisk2015privacy} with respect to sensor data. Second, in the absence of legal obligation to participate, many critical utility providers do not pursue membership of data sharing cooperatives like CRISP  due to a variety of business as well as logistical reasons. Third, programs like CRISP also require hardware upgrades to existing IT infrastructure and associated costs as well \cite{multiYearPlan}. 
Lastly, stakeholders might harbor inherent distrust in the centralized coordinator itself which is usually the federal government agencies and/or regulatory bodies. Such distrust could be due to apprehensions pertaining to unintended use of sensor data by regulators to detect violations of critical infrastructure rules and impose penalties. 

The unique characteristics of the blockchain could be leveraged to achieve the best of both worlds: the privacy and agility of a purely localized approach coupled with the accuracy and robustness of the centralized method. The blockchain could allow stakeholders to use state of the art ICS attack detection mechanisms locally while being able to share their alarm statistics with all other stakeholders in real time. Further, the underlying consensus mechanism of the blockchain would be useful in estimating the global cyber health status by pooling local alarms in a fully trustworthy and transparent fashion. While the blockchain provides a robust computational platform for global aggregation of alarms, the local ICS attack detection algorithms would eliminate the need to move low level sensor data ensuring full data privacy. As a result, stakeholders would be more amenable to participating in such a decentralized mechanism that provides prompt insights while potentially reducing false alarm rates drastically. The significant reduction in setup and maintenance costs due to elimination of a centralized aggregator \cite{ramanan2019baffle} is an understated advantage of a blockchain driven approach.}

In this paper, we propose the use of a permissioned blockchain for detecting globally coordinated replay attacks in a decentralized fashion. Blockchains typically rely on a consensus among multiple mistrusting parties to achieve a consistent global state. As a result, blockchain based platforms are decentralized in nature and do not involve a centralized chain of command. Our framework therefore ensures data privacy by allowing individual utilities to run their own detection algorithms locally. This leads to full ownership of cyber incident sensor data with complete privacy. It also eliminates the associated cost of setting up a trusted, centralized third party. Moreover, we utilize blockchain driven Smart Contracts (SC) \cite{wang2019blockchain} to estimate the likelihood of a global replay attack based on alarms and insights aggregated from the various utilities. 

In order to demonstrate the blockchain's ability to rapidly generate and propagate global insights, we perform an in depth comparison with well established information diffusion mechanisms like Broadcast Gossip (BG). Gossip protocols are an important type of diffusion technique aimed at estimating the global state of a system in a peer-to-peer fashion 
\cite{boyd2006randomized,bg}. As a result, in this paper, we develop a BG based global replay attack detection detection framework for benchmarking purposes. We compare and contrast our blockchain driven approach against the BG based framework in a theoretical as well as in an empirical manner. Our results indicate that the blockchain provides a sound computational platform allowing for global aggregation of local outputs in a timely, accurate and reliable fashion. The major contributions of this work are summarized in detail below:
\begin{itemize}
    \item We develop a decentralized mechanism that relies on Bayesian inference in order to detect a globally coordinated replay attack with full regional data privacy.
    \item We introduce Theorem \ref{thm1}, specifically geared towards maintaining computational efficiency of Bayesian inference on a blockchain platform. 
    \item With the help of Theorem \ref{thm1}, we design a blockchain based framework for computing the Global Attack Probability (GAP) with only one global multiplication and addition steps.
    \item We propose the BG framework as a benchmark and theoretically compare its performance with the blockchain based approach with the help of Theorem \ref{thm2}.
    \item We implement our framework on an Ethereum based private blockchain network to demonstrate its scalability and applicability for varying degrees of cyber threat parameters.
\end{itemize}
Our results conclusively show that the blockchain driven framework is vastly superior to conventional, state-of-the-art information diffusion paradigms, like BG, both in terms of computational performance as well as accuracy of results.
\section{Related Work}
{
\textcolor{black}{
ICSs are a special class of information systems, involving interaction between information technology (IT) and operational technology (OT)\cite{cook2017industrial}. Like all information systems, cyberattacks also threaten many ICSs involved in applications such as water supply, manufacturing, power systems, and energy \cite{sturm2017cyber,kosut2011malicious,copeland2010terrorism}. Attacks like DoS, DDoS, phishing, spoofing and eavesdropping that target generic IT systems can also target ICSs but they can often be effectively detected and isolated by monitoring network traffic  \cite{lazarevic2003comparative,caselli2016specification,ye2004robustness}. On the other hand, data attacks are a more significant threat to ICSs due to their interaction between information, communication and the underlying physical processes \cite{urbina2016limiting} that could potentially lead to critical infrastructure damage.}

\subsection{Data Manipulation Attack Detection}
In data manipulation attacks, attackers manipulate controller/sensor data or even data at rest in order to damage critical assets through malicious control actions and incorrect state estimations leading to degraded asset performance. It has been shown that data attacks can be designed to bypass basic verification methods relying on Cyclic Redundancy Check (CRC), User Datagram Protocol (UDP) and Transmission Control Protocol (TCP). \cite{kang2014cyber,drias2015taxonomy}.

Numerous model-based detection frameworks have been proposed as an added layer of protection for ICS attacks involving data manipulation \cite{liu2011false,dan2010stealth, teixeira2010cyber, yu2015blind, cardenas2011attacks,chaojun2015detecting,mo2014detecting, huang2018online,van2015sequential,hoehn2016detection,weerakkody2015detecting,pan2015classification,rahman2017multi}.
Most detection algorithms rely on differences between actual measurements and those estimated by a model of the physical system. Differences between the estimated (or predicted) and observed measurements or states (i.e., residuals) can be used to detect possible cyberattacks. In most cases, a sequential goodness-of-fit testing procedure serves as the basis for detecting the attack. However, these approaches have been shown to be inefficient in detecting replay attacks, primarily because the observed measurements (replayed data) often match measurements estimated by the system's model. }

{Replay attacks can be detected based on several types of strategies. The first type is watermarking the control actions \cite{romagnoli2019model, mo2014detecting, huang2018online,porter2020detecting,hespanhol2020sensor, mo2015physical}. This type of strategy is based on the assumption that the operator can add watermarks (a small bias or noise) to the control actions, which is unknown to the attacker. The second type is observer-based anomaly detection \cite{irita2017detection,danLi}. This type of strategy is based on the assumption that some of the sensors can be protected or immune from manipulation. The local attack detection algorithm we adopt in this paper \cite{danLi} belongs to the second type, however, it does not assume the operators can proactively protect the sensors, but that some certain types of sensors are immune to manipulation. In fact, any replay attack detection technique can be adapted to our global framework, as long as the detection algorithm can be used to obtain the probability of attack. The advantage of adopting the method in \cite{danLi} is that it detects the attack based on statistical inference, which naturally outputs the probability of attack.}

\subsection{Blockchain Driven Methods for Security and Privacy}
The use of blockchain has been proposed as a gateway to ensuring data privacy and security in a wide variety of application areas. A lightweight, private blockchain paradigm for enhancing the security and privacy of IIoT driven manufacturing platforms has been proposed in \cite{smartfac}. Further, \cite{zhou2019secure} provides a secure, permissioned blockchain based approach for energy trading between electric vehicles and the grid. 
Use of blockchain has been considered in power systems in numerous works in the recent past. Blockchains have been proposed as a means to establish a decentralized, secure technique for transactive energy \cite{BC_Resilient} as well as to handle and trace back energy losses in microgrids that incorporate PV nodes \cite{enmg}. Blockchain driven approaches are also being considered as the perfect computational platforms for collaboratively detecting attacks and anomalies on a global scale with full data privacy \cite{bc_cids}. However, despite its immense potential, efforts exploring the use of the blockchain towards detecting globally coordinated replay attacks on power network ICSs largely remains understated.   
\begin{table*}[!htb]
\centering
\captionsetup{margin={0.1\textwidth,0pt},singlelinecheck=false}
\caption{{ Comparative Analysis of State-of-the-art Consensus Protocols for Permissioned Blockchains}}
\label{tab:poa_comparison}
{ \begin{tabular}{c|c|c|c|c|c|c}
 
\textbf{protocol} & \textbf{PoA type} & \textbf{validators} & \textbf{decentralization} & \textbf{fault tolerance} & \textbf{consistency} & \textbf{block latency}\\
\hline
Clique & modern BFT & at least 1 & high &  $<$ 1/2 validators & \makecell{eventual \\ consistency} & \makecell{low; only 1 message \\ round required\\} \\
& & & & & \\
IBFT & classical BFT & at least 4 & medium & $<$ 1/3 validators & \makecell{consistent but \\ stalling likely} & \makecell{high; 3 message \\ rounds required\\} \\
& & & & & \\
Aura & modern BFT & at least 1 & high & $<$ 1/2 validators & \makecell{no consistency \\ guarantees} & \makecell{medium; 2 message \\ rounds required\\} \\
& & & & & \\
Raft & non BFT & \makecell{exactly 1 leader, \\multiple verifiers, \\ many learners} & \makecell{low; restricted \\to leader} & $<$ 1/2 crash faults & \makecell{consistent \\ (assumes no forking)} & \makecell{low; only leader \\ can mint blocks} \\
\end{tabular}}
\end{table*}

{ 
\subsection{Operational Dynamics of Blockchains}
The operational dynamics of a blockchain implementation are dictated by their individual application or use case. In cases where no assumptions can be made about the identity or intent of the participating entities, a permissionless blockchain implementation is the norm \cite{POA_analysis}. On the other hand, in applications where the parties are well-established real-world entities which elicit a fair degree of trust, a permissioned blockchain implementation is preferred \cite{salimitari2018survey}. The major distinguishing factor between a permissioned and a permissionless blockchain is the underlying consensus protocol which ensures a consistent global state among all parties. 
Permissionless blockchain implementations involve computationally intensive consensus protocols that are designed to operate without any assumption of trust among the mistrusting parties. Popular consensus mechanisms for permissionless blockchains include Proof-of-Work (PoW), Proof-of-Stake (PoS), delegated Proof-of-Stake (dPoS), Proof-of-ElapsedTime (PoET) \cite{salimitari2018survey}. However, permissioned blockchains place an emphasis on high performance involving consensus protocols that leverage the identity of participating entities. In permissioned blockchains, consensus protocols are primarily a fault-tolerant means of achieving global consistency by inducing a tamper proof record \cite{salimitari2018survey}.

\subsection{Consensus Mechanisms for Permissioned Blockchains}
Consensus mechanisms in permissioned settings are usually dominated by Proof-of-Authority (PoA) based algorithms. Rather than relying on solutions to complex mathematical-based challenges like PoW, PoA relies on the authority of real-world entities part of a permissioned ledger. A majority of the authorities have to achieve consensus before a block is permanently added to the chain. As a result, PoA improves security without relying on computational challenges, since an attacker must hack a majority of authorities in order to rescind all transactions \cite{poachains}. PoA chains are also known to have low latency, deterministic block creation process, as well as faster block creation times which are more important in a permissioned setting \cite{poachains,salimitari2018survey}. PoA is more suited to a consortium or permissioned settings due to its security, easier maintenance as well as accountability of the authorities themselves \cite{poachains}. Therefore, in this paper, we adopt a permissioned blockchain driven approach since our problem setting involves a \textit{consortium of real-world entities} such as utility stakeholders with their own ICSs desirous of detecting a global replay attack. 

Most PoA algorithms are based on the concept of Byzantine fault tolerance (BFT) \cite{POA_analysis, salimitari2018survey}. However, non-BFT algorithms like Raft \cite{salimitari2018survey} have also been proposed and form an integral part of many permissioned blockchain implementations such as Corda and GoQuorum. The class of BFT driven PoA mechanisms can in turn be divided into two categories: classical and modern. The classical PoA category consists of protocols such as Istanbul BFT \cite{saltini2019ibft} (offered in Hyperledger Besu, and GoQuorum) with the consistency of the ledger as the priority. The modern category consists of a new family of BFT protocols such as Clique (offered in Hyperledger Besu, Ethereum, OpenEthereum and GoQuorum) and Aura \cite{POA_analysis} (offered in OpenEthereum). The modern PoA protocols emphasize better performance (i.e. low latency and block creation times) over consistency. 


We summarize the characteristic features of widely used consensus protocols i.e. Clique, IBFT, Aura and Raft with the help of Table \ref{tab:poa_comparison}.  Table \ref{tab:poa_comparison} provides a qualitative comparative analysis of each protocol. For an in-depth discussion of the characteristic features presented in Table \ref{tab:poa_comparison}, we refer the reader to Appendix \ref{sec:appendixA}. Due to its low latency, high decentralization potential, increased degree of fault tolerance as well as the capability to achieve eventual consistency, we choose the Clique consensus protocol for our experiments.}




\section{Problem Formulation}\label{sec:probForm}
We primarily use the blockchain as a computational platform for aggregating outputs from local replay attack detection algorithms. In developing our blockchain based global attack detection framework, we consider a power network that is divided topologically into a set of distinct regions denoted by $\mathcal{R} = \{1,2,3 \ldots n\}$. Each region can be thought of as a utility provider with multiple power plants. A blockchain architecture is used to assess the global state of the entire network based on insights that are aggregated in a decentralized, privacy preserving manner. The aggregation determines the overall probability of a global network attack.

We first consider a single region $i$ comprised of $p$ generators monitored and controlled by a single ICS where the state of each generator can be represented by $m$ variables. We assume that the dynamics of this system of generators can be characterized by a linear time-invariant (LTI) model at time $t$ described by Equations \eqref{eq:eqn1a}, \eqref{eq:eqn1b}. 
\begin{gather}
    x_{t+1}=Ax_t+Bu_t+v_t,\label{eq:eqn1a}\\
    y_t=Cx_t+w_t, \label{eq:eqn1b}
\end{gather}
In Equations \eqref{eq:eqn1a}, \eqref{eq:eqn1b}, $x_t\in \R^{mp}$ represents the unobserved state of the system, \textcolor{black}{$u_t\in \R^{mp}$} is the control action while $y_t\in \R^{mp}$ denotes the sensor 
measurements, which are assumed to be noisy realizations of the state. \textcolor{black}{$A\in \R^{mp\times mp}$} is the state transition matrix and represents the system dynamics. \textcolor{black}{$B \in \R^{mp\times mp}$} is the input matrix and represents how the control action impacts the system state. $C$ is the measurement matrix and represents the relationship between the measurements $y_t$ and the state $x_t$. $v_t\in \R^{mp}$ and $w_t\in \R^{mp}$ represent process and measurement noise at time $t$, respectively. Such a type of modeling framework has also been used extensively in prior art \cite{smith2011decoupled, mo2014detecting, teixeira2010cyber}.

In this setting, the Kalman filter is known to be the optimal state estimator\cite{Anderson}. The residuals $r_t$ are defined as the differences between the actual measurements $y_t$ and the predicted measurements $C\hat{x}_{t|t-1}$ as estimated by the Kalman filter. The control action $u_t$ is calculated using a linear-quadratic Gaussian (LQG) controller, which is the optimal controller under the LTI setting \cite{LQG}. The Kalman filter and the LQG controller are estimated based on Equations \eqref{eq:KLQG1}-\eqref{eq:KLQG4}.
      \begin{gather}
            \hat{x}_{t|t-1}=A\hat{x}_{t-1|t-1}+Bu_{t-1} \label{eq:KLQG1},\\
		    \hat{x}_{t|t}=\hat{x}_{t|t-1}+Kr_t \label{eq:KLQG2},\\
		    r_t=y_t-C\hat{x}_{t|t-1} \label{eq:KLQG3},\\
		    u_t=L\hat{x}_{t|t}. \label{eq:KLQG4}
    	\end{gather}
\textcolor{black}{In Equation \eqref{eq:KLQG4}, $L=-(B^TSB+U)^{-1} B^T S A$ where $S$ is the solution to the Ricatti equation \cite{LQG} denoted by Equation \eqref{eq:ricatti}.
\begin{equation}\label{eq:ricatti}
  S=A^T SA+W-A^T SB (B^T SB+U)^{-1}B^TSA   
\end{equation}
In Equation \eqref{eq:ricatti}, $W,U \in \R^{mp\times mp}$ are positive semidefinite matrices used by the LQG controller to minimize its objective function $J = E[x_t^TWx_t + u_t^TUu_t]$ based on state variables and control actions, respectively. Therefore, $W,U$ primarily denote the degree of aggression while regulating the state variables.} 



{ 
\subsection{Threat Model}
As mentioned in Section I, in a replay attack, the attacker gains access to both sensors and controller. This can be achieved by intruding into the ICS through its supervisory computers or other connected hardware. After the access is obtained, the attack is implemented in two stages. The first stage involves eavesdropping and recording of sensor measurements $\{y^r_t\}$ without data manipulation, which can be denoted by
\begin{equation}
  y^r_t=y_t.  
\end{equation}
The second stage is the manipulation of the control actions $u_k$ by injecting bias or noise, while replaying the sensor measurements $y^r_t$ in replacement of $y_t$. The second stage is represented by Equations (\ref{eq:ra1}) and (\ref{eq:ra2})
\begin{gather}
    u'_t=u_t+a_t \label{eq:ra1},\\
    y'_t=y^r_{t-\delta},\label{eq:ra2}
\end{gather}
where $u'_t$ and $y'_t$ are the control action and measurements under attack, respectively, $a_t$ is the bias or noise injected at time $t$, and $\delta$ is the time difference between the onset of the two stages of the replay attack. In this way, the system state is altered due to Equation \eqref{eq:ra1}, but the attack is hard to detect by monitoring the sensor data due to Equation \eqref{eq:ra2}.


In our framework, we assume that the adversary is outside the blockchain and is solely focused on injecting data into a utility stakeholder's ICS in order to carry out a replay attack. In our framework, the stakeholders themselves trust each other and share the insights from their local ICS detection algorithms with each other through the permissioned blockchain. Furthermore, we also assume that an attacker targets at most less than 50\% of the regional utility stakeholders participating in the consortium. Therefore, based on the PoA protocol \cite{POA_analysis}, we can state that an attacker controlling hacked authorities cannot revert transactions and attack the network. 
We assume that all regions are equally likely to experience a replay attack which are mutually independent. Independence is assumed here for mathematical convenience. More importantly, a global network attack is assumed to occur when at least two regions report local attacks. 

Given our threat model, the problem formulation can be viewed as having two components, a local and a global component. The local component represents regional plants that belong to a single utility provider. We assume that each region executes the local algorithms aimed at detecting replay attacks on their plants' ICS. We refer to the local component as the \emph{regional detection model}. The global component, i.e. the \textit{network detection model}, concerns the detection of a coordinated global attack at the network level.}

\subsection{Regional Detection Model}\label{my:local}

As shown in \cite{danLi}, replay attacks can be detected by monitoring the covariance matrix of the residuals $r_t$. Using this approach, let $\sigma^i$ be a random variable representing the probability of an alarm triggered by region $i$. That is, $\sigma^i=1$ if a replay attack is detected and $\sigma^i=0$ otherwise. We define $\hat{\sigma}^i$ as the ground truth that represents whether a replay attack is indeed underway. There are two classic errors that can occur in this setting. Type-I error, $\alpha$, represents the probability of a false alarm, i.e., the algorithm triggers an alarm when there is no attack. The Type-II error, $\beta$, represents the probability of a false negative, i.e., where the detection algorithm fails to detect a true replay attack. These errors can be defined more formally in terms of a region $i$ as follows 
\begin{gather}
\alpha_i = Pr(\sigma^i=1|\hat{\sigma}^i=0), \label{eq:al}\\
 \beta_i = Pr(\sigma^i=0|\hat{\sigma}^i=1). \label{eq:be}
\end{gather}
In Equations \eqref{eq:al}, \eqref{eq:be} $\hat{\sigma}^i=1$ in the event the system is truly under a replay attack whereas $\hat{\sigma}^i=0$ otherwise.
\subsection{Network Detection Model}
Consider a power network comprised of $n$ regions each reporting an alarm based on their local belief of an attack. Recall that a global network attack is triggered if there are two or more distinct regional alarms, i.e. at least two regions detect an attack. Let the set $S = \{\bm{s}_0,\bm{s}_1 \ldots \bm{s}_n\}$ denote all the scenarios that represent no global network attack, where $\bm{s}_0 \in \{0\}^n$ indicates a scenario where no regional alarms have been triggered and $\bm{s}_i \in \{0,1\}^n$ denotes scenarios where only one region $i$ triggers an alarm.
Consequently, for a scenario $\bm{s}_k,$ $k>0$, only its $k^{th}$ element $s^k_k = 1$. Next, let $\bm{\sigma} \in \{0,1\}^n = [\sigma^1,\sigma^2\ldots \sigma^n]$ be a vector of random variables representing regional alarm events. Given $\bm{\sigma}$, the probability of no global network attack can be expressed as $Pr(S|\bm{\sigma})$. Similarly, $Pr(\bm{\sigma}|\bm{s}_k)$ defines the probability of observing $\bm{\sigma}$ given a scenario $\bm{s}_k$
Since regional attacks are assumed to be independent, the following expressions hold true: \eqref{eq:a1}-\eqref{eq:a4}.
\begin{gather}
    Pr(\bm{s}_k) = \prod\limits_{i=1}^nPr(s^i_k), \label{eq:a1}\\
    Pr(\bm{\sigma}|\bm{s}_k) =\prod_{i=1}^{n}Pr(\sigma^i| s^i_k),\label{eq:a2}\\
    Pr(S|\bm{\sigma}) = \sum_{\bm{s}_k \in S}  \frac{Pr(\bm{\sigma}|\bm{s}_k)Pr(\bm{s}_k)}{Pr(\bm{\sigma})}, \label{eq:a3} \\
    Pr(\bm{\sigma}) = Pr(\sigma^1 \cap \sigma^2 \ldots \cap\sigma^n) = \prod\limits_{i=1}^{n}Pr(\sigma^i). \label{eq:a4}
\end{gather}
The occurrence of the event denoted by $s^i_k$ is only relevant in the context of the ground truth $\hat{\sigma}^i$. Therefore, we can state that $Pr(s^i_k) = Pr(\hat{\sigma}^i = s^i_k)$, where $Pr(\hat{\sigma}^i = s^i_k)$ reflects the prior probability of the existing ground truth $\hat{\sigma}^i$ being equal to $s^i_k$. 
Consequentially, substituting Equations \eqref{eq:a1} and \eqref{eq:a2} in \eqref{eq:a3}, we obtain the relation represented by \eqref{eq:sumP}.
\begin{equation}\label{eq:sumP}
    \sum_{\bm{s}_k \in S} Pr(\bm{\sigma}|\bm{s}_k)Pr(\bm{s}_k) =\sum_{\bm{s}_k \in S}\prod_{i=1}^nPr(\sigma_i|\hat{\sigma}^i=s^i_k)Pr(\hat{\sigma}^i=s^i_k).
\end{equation}
We also know that
\begin{equation}\label{eq:sigma}
    Pr(\sigma^i) = Pr(\sigma^i|\hat{\sigma}^i=0)Pr(\hat{\sigma}^i=0) + Pr(\sigma^i|\hat{\sigma}^i=1)Pr(\hat{\sigma}^i=1).
\end{equation}
We note that $Pr(\sigma^i)$ and $Pr(\sigma_i|s^i_k)Pr(\hat{\sigma}^i=s^i_k)$ in Equations \eqref{eq:sumP}, \eqref{eq:sigma}, respectively, can be computed locally allowing us to propose the the following theorem.
\begin{theorem} \label{thm1}
The probability of no global attack denoted by $Pr(S|\sigma)$ is governed by the equality,
\[Pr(S|\bm{\sigma}) = \Big(\prod_{i=1}^n\frac{a_i}{Pr(\sigma^i)}\Big)\Big(1+\sum_{i=1}^n\frac{b_i}{a_i}\Big)\]
where, \[a_i = Pr(\sigma^i|\hat{\sigma}^i=0)Pr(\hat{\sigma}^i=0),\text{ }b_i = Pr(\sigma^i|\hat{\sigma}^i=1)Pr(\hat{\sigma}^i =1)\]
\end{theorem}
\begin{proof}
Based on our assumptions, we derive the relations represented by Equations \eqref{eq:main2}, \eqref{eq:main3}.
\begin{gather}
    Pr(\bm{\sigma}|\bm{s}_0)Pr(\bm{s}_0) = \prod_{i=1}^{n}Pr(\sigma^i|\hat{\sigma}^i=0)Pr(\hat{\sigma}^i=0),\label{eq:main2}\\
    Pr(\bm{\sigma}|\bm{s}_k)Pr(\bm{s}_k) = \prod_{i=1}^{n}Pr(\sigma^i|\hat{\sigma}^i=s^i_k)Pr(\hat{\sigma}^i=s^i_k)\label{eq:main3}.
\end{gather}
Dividing Equations \eqref{eq:main3} with \eqref{eq:main2}, we get
\begin{equation}\label{eq:frac}
    \psi_k=\frac{Pr(\bm{\sigma}|\bm{s}_k)Pr(\bm{s}_k)}{Pr(\bm{\sigma}|\bm{s}_0)Pr(\bm{s}_0)} = \prod_{i=1}^{n}\frac{Pr(\sigma^i|\hat{\sigma}^i=s^i_k)Pr(\hat{\sigma}^i=s^i_k)}{Pr(\sigma^i|\hat{\sigma}^i=0)Pr(\hat{\sigma}^i=0)}.
\end{equation}
Since, $s^i_k = 1$ only when $i=k$, we can rewrite Equation \eqref{eq:frac} as
\begin{equation}\label{eq:frac2}
    \psi_k
    =\Big[\prod_{\substack{i=1\\i\neq k}}^{n}\frac{Pr(\sigma^i|\hat{\sigma}^i=0)Pr(\hat{\sigma}^i=0)}{Pr(\sigma^i|\hat{\sigma}^i=0)Pr(\hat{\sigma}^i=0)}\Big] \frac{Pr(\sigma^k|\hat{\sigma}^k=1)Pr(\hat{\sigma}^k=1)}{Pr(\sigma^k|\hat{\sigma}^k=0)Pr(\hat{\sigma}^k=0)}.
\end{equation}
Equations \eqref{eq:frac}, \eqref{eq:frac2} imply that,
\begin{equation}\label{eq:sub1}
    Pr(\bm{\sigma}|\bm{s}_k)Pr(\bm{s}_k) = Pr(\bm{\sigma}|\bm{s}_0)Pr(\bm{s}_0)\frac{b_k}{a_k},\quad \forall k>0 .
\end{equation}
Recall from Equation \ref{eq:sumP}, that
\begin{equation}\label{eq:sub2}
    Pr(\bm{\sigma}|\bm{s}_0)Pr(\bm{s}_0) = \prod_{i=1}^nPr(\sigma^i|\hat{\sigma}^i=0)Pr(\hat{\sigma}^i=0) = \prod_{i=1}^na_i.
\end{equation}
Since $a_k,b_k$ can be purely computed by region $k$ we obtain Equation \eqref{eq:penultimate} by combining Equations \eqref{eq:sub1} and \eqref{eq:sub2} 
\begin{equation}\label{eq:penultimate}
    \sum\limits_{\bm{s}_k \in S} Pr(\bm{\sigma}|\bm{s}_k)Pr(\bm{s}_k)= \Big(\prod_{i=1}^na_i\Big)\Big(1+\sum_{i=1}^n\frac{b_i}{a_i}\Big).
\end{equation}
Based on  Equations \eqref{eq:penultimate} and \eqref{eq:a4}, we obtain
\[Pr(S|\bm{\sigma}) = \frac{\sum\limits_{\bm{s}_k \in S} Pr(\bm{\sigma}|\bm{s}_k)Pr(\bm{s}_k)}{\prod\limits_{i=1}^{n}Pr(\sigma^i)} = \Big(\prod_{i=1}^n\frac{a_i}{Pr(\sigma^i)}\Big)\Big(1+\sum_{i=1}^n\frac{b_i}{a_i}\Big)\].
\end{proof}
The entities $a_i, b_i$ can be computed in a simplified manner as follows
\begin{gather}
    a_i = Pr(\sigma^i,\hat{\sigma}^i=0) = Pr(\sigma^i|\hat{\sigma}^i=0)Pr(\hat{\sigma}^i=0), \label{eq:abcalc1}\\
    b_i = Pr(\sigma^i,\hat{\sigma}^i=1) = Pr(\sigma^i|\hat{\sigma}^i=1)Pr(\hat{\sigma}^i=1). \label{eq:abcalc2}
\end{gather}
Moreover, the prior distribution can be updated in a purely local fashion using:
\begin{gather}
    Pr(\hat{\sigma}^i=0) = \frac{a_i}{a_i+b_i}, \label{eq:p1}\\
    Pr(\hat{\sigma}^i=1) = \frac{b_i}{a_i+b_i}. \label{eq:p2}
\end{gather}
Theorem \ref{thm1} indicates that the global attack probability can be computed through one global multiplication and addition of the two terms, $\frac{a_i}{Pr(\sigma^i)},\frac{b_i}{a_i}$. We therefore incur significant computational benefits especially on a blockchain based framework where computation is expensive.
\begin{figure*}[t!]
\centering
\captionsetup{justification=centering}
\includegraphics[width=\textwidth,keepaspectratio]{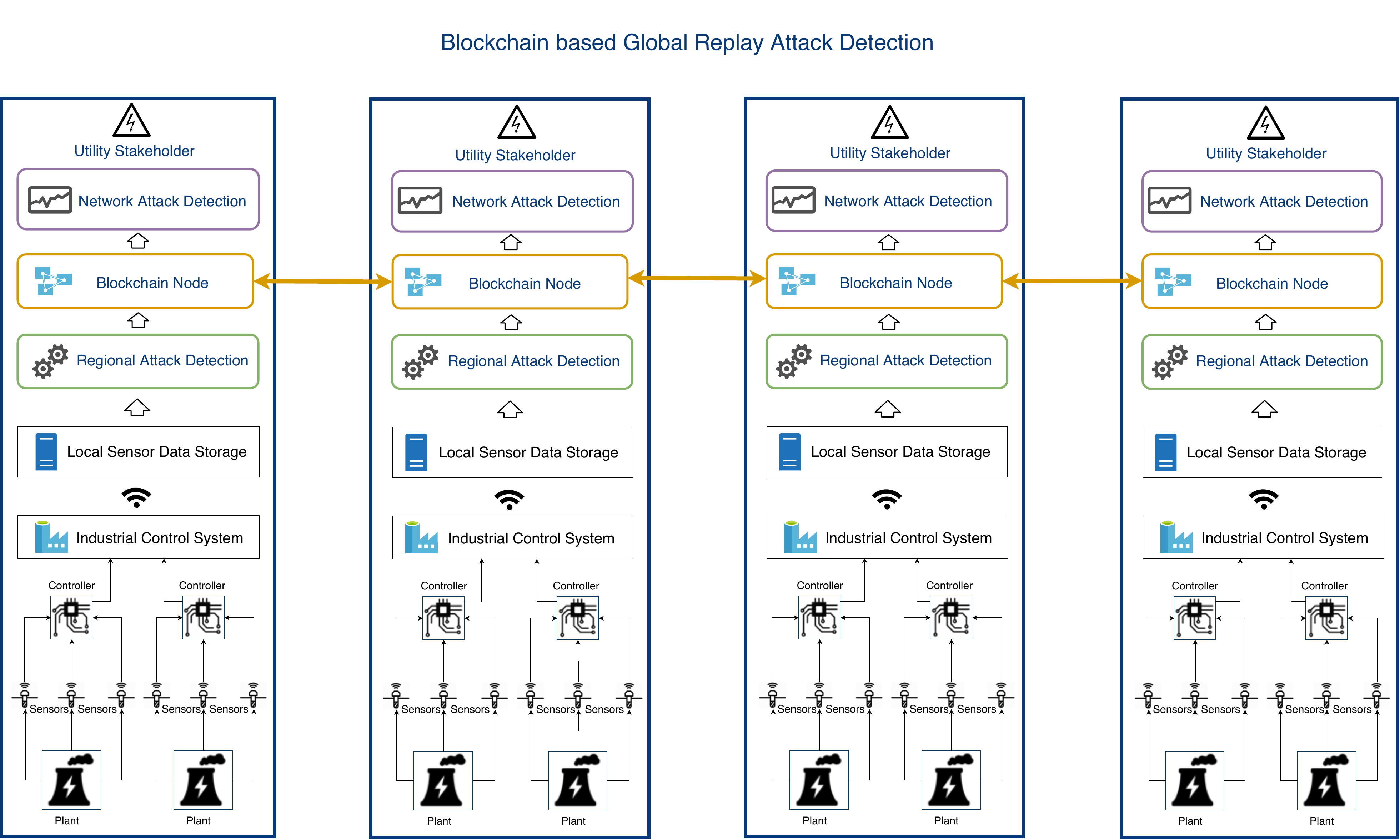}
\caption{System Architecture for blockchain based global replay attack detection}
\label{fig:sysarch}
\end{figure*}
Figure \ref{fig:sysarch} represents the system architecture for our blockchain based global attack detection mechanism. At the lowest level, we have plants and their associated sensors feeding data into controllers that are directly managed by the Industrial Control System (ICS). The ICS stores the data in a local data store on which the regional detection algorithm is applied. Insights from the regional detection model are shared using Theorem \ref{thm1} on the blockchain. Every regional utility stakeholder utilizes the latest available information on the ledger to be able to obtain knowledge of a global replay attack. 

\section{Blockchain Based Framework for Global Replay Attack Detection}%
The introduction of Smart Contracts
(SCs)\cite{wang2019blockchain,zheng2020overview} have been an important addition to the blockchain paradigm. An SC is typically a snippet of code that resides on the blockchain. It can contain complex program logic on the blockchain and can be invoked by any party having access to the blockchain. Once invoked, an SC is self-triggering and proceeds to alter the state of the ledger with the help of the underlying consensus protocol. 
As a result, an SC can be used for executing business logic through consensus among mistrusting parties paving the way for a decentralized application. Therefore, due to their versatility, blockchain driven Smart Contracts (SC) provide an ideal environment for the Network Detection model in a fully decentralized fashion. 

Solidity \footnote{\url{https://github.com/ethereum/solidity}} is a popular SC oriented programming language that can be leveraged for developing highly versatile decentralized applications. A key constraint of languages like Solidity is the lack of floating point arithmetic operations \footnote{\url{https://solidity.readthedocs.io/en/v0.5.3/types.html\#fixed-point-numbers}}. This is done to primarily reduce the computational burden on the underlying consensus protocols. As a result, one has to define a precision conversion factor for converting floating point values to integers before feeding them as inputs to the SC on the front end side. Obviously, the choice of the precision factor could have a wide impact on the detection accuracy of a global replay attack. In this section, we propose a Solidity based Smart Contract(SC) design that embodies our global replay attack detection paradigm discussed in Section \ref{sec:probForm}.

\subsection{Smart Contract Design}
Table \ref{tab-scfeat}, \ref{tab-scfunc} depict the attributes and functions that form an integral part of the SC design for our framework. In Table \ref{tab-scfeat}, the precision factor is denoted by $D$ and arrays $aOPs,bOa$ store the integer values $\frac{D a_i}{Pr(\sigma^i)},\frac{D b_i}{a_i}$ of all regions respectively. In Table \ref{tab-scfunc}, the functions $updateData$, $aggregateValues$ can be asynchronously invoked by any region for updating the SC with their local values and obtaining the corresponding global aggregates respectively. 
\begin{table}[!htb]
\centering
\caption{Solidity based Smart Contract Attributes}
\label{tab-scfeat}
\begin{tabular}{c|c|c}
\toprule
\textbf{Attribute}  & \textbf{Type} & \textbf{Description}      \\ \hline
\textit{D} & \textit{uint256} & Precision for floating point conversion.\\
\textit{n} & \textit{uint} & The total number of regions i.e. $|\mathcal{R}|$\\
\textit{aOPs}[] & \textit{uint256} & Array for storing $\frac{D a_i}{Pr(\sigma^i)},\forall i\in\mathcal{R}$\\
\textit{bOa}[] & \textit{uint256} & Array for storing $\frac{D b_i}{a_i},\forall i \in \mathcal{R}$\\
\bottomrule
\end{tabular}
\end{table}
\begin{table}[!htb]
\centering
\caption{Solidity based Smart Contract Function}
\label{tab-scfunc}
\begin{tabular}{c|c|c}
\toprule
\textbf{Function}  &\textbf{Invoker} & \textbf{Description}      \\ \hline
\textit{updateData} & \textit{Region i} & sets $aOPs[i]$, $bOa[i]$\\ \textit{aggregateValues} & \textit{Region i} & returns
$\prod\limits_{j=1}^n aOPs[i]$, $\sum\limits_{j=1}^n bOa[i]$\\
\bottomrule
\end{tabular}
\end{table}
\subsection{Blockchain Based Global Attack Detection Algorithm}
Algorithm \ref{alg:bdpgat} presents the details of blockchain based global replay attack detection which is executed in a decentralized fashion. We assume an off chain interaction informs all the regions about the SC address, total number of regions as well as the precision factor. Every region initially determines its corresponding local alarm value $\sigma^i$. Let $x_i=\frac{a_i}{Pr(\sigma^i)}, y_i=\frac{b_i}{a_i}$ represent the local statistical values which are converted to integers using the precision factor $D$ and pushed to the SC by invoking $SC.updateData$. Any region can asynchronously invoke $SC.aggregateValues$ in order to obtain $x^b=\prod\limits_{i=1}^nD x_i$ and $y^b=\sum\limits_{i=1}^n D y_i$. A global estimate of no attack can then be computed locally using Theorem \ref{thm1}. A complement of the result can be used to infer the presence of a global attack. At each epoch, the prior distribution pertaining to local alarm values gets updated locally.
\begin{algorithm}
\caption{Decentralized blockchain based algorithm}\label{alg:syncd}
\label{alg:bdpgat}
\begin{algorithmic}
\State owner region deploys and initializes $SC$ on blockchain
\For{i=1,2,3\ldots n \textbf{in parallel}}
\State initialize $D,n$ and obtain $SC$ address.
\While{true}
\State determine the message $\sigma^i$
\State compute $Pr(\sigma^i)$ using Equation \eqref{eq:sigma}
\State compute $x_i,y_i$ using Equations \eqref{eq:abcalc1},\eqref{eq:abcalc2}
\State invoke $SC.updateData\Big(D x_i,D y_i,i\Big)$ 
\State $x^b,y^b\leftarrow SC.aggregateValues()$
\State using Theorem \ref{thm1} compute $1-Pr(S|\bm{\sigma})$ such that, \State \ \ \qquad \qquad \qquad \qquad $Pr(S|\bm{\sigma})=x^b(D+y^b)/D^{(n+1)}$ 
\State update prior distribution using Equations \eqref{eq:p1},\eqref{eq:p2} 
\EndWhile
\EndFor
\end{algorithmic}
\end{algorithm}
\section{Performance Comparison of Blockchain and BG driven approaches}
For developing a decentralized global replay attack detection, we assume the existence of an underlying connectivity graph that represents the system level interconnection among utilities \cite{rosas2007topological}. In this graph, each vertex represents a utility or region. An edge exists if there is a shared transmission line between the corresponding utilities. Without a central aggregator, a diffusion mechanism must take place over the existing connectivity graph in order to detect globally coordinated replay attacks. Therefore, we develop a novel diffusion algorithm based on state-of-the-art BG as a benchmark strategy for our blockchain based approach \cite{bg}. BG can be used to compute the global average of values held by the vertices of the connectivity graph. BG avoids the computational bottlenecks found in other gossip protocols  while converging to the global average in expectation \cite{bg}. Being inherently decentralized, consensus driven and peer-to-peer in nature, BG forms the ideal benchmark for comparing the performance of a blockchain driven framework. 
\subsection{BG based Reformulation of Global Attack Detection}
Since the total number of agents in the system is known, BG can be used to estimate the global sum and product terms present in Theorem \ref{thm1} as well.  Recall $x_i,y_i$, $\forall i \in \mathcal{R}$ from Algorithm \ref{alg:bdpgat}. For computing the sum $y$, the reformulation is trivial and consists of estimating the global sum from the global average $\sum\limits_{k=1}^ny_i/n$. In order to use BG to calculate the product $x=\prod\limits_{i=1}^nx_i$, we reformulate $x = e^u$, where $u=\sum\limits_{i=1}^nu_i$, and $u_i=\log{x_i}$. We leverage the BG protocol to compute $\sum\limits_{i=1}^nu_i/n$ which can be used to estimate $x$ in a purely peer-to-peer fashion.
\subsection{Performance Analysis}
We wish to compare and contrast the precision factor based error introduced in the blockchain framework against the expected asymptotic error in a BG driven framework. As a result, we propose Theorem \ref{thm2} which helps characterize the conditions favorable for the BG to remain competitive with our blockchain framework. For Theorem \ref{thm2}, we consider a set of agents connected by graph $G$ with Laplacian $L$. For our analysis of Theorem \ref{thm2}, we assume a constant value of $\beta$ across all regions. Further, we let $\lambda_{n-2}(L)$ and $\lambda_{1}(L)$ denote the second smallest and the largest eigenvalues of $L$ where $\gamma\in(0,1)$ represents the mixing parameter \cite{bg}. The mixing parameter dictates the contribution of values recieved from each neighbor during the BG protocol at each vertex \cite{bg}.
As illustrated in \cite{bg}, the expected BG error bound for computing the average $\bar{z} = \sum\limits_{i=1}^nz_i/n$ where $z_i$ is the local value held by agent $i \in \{1,2,\ldots n\}$ is given by
\begin{equation}\label{eq:bgerr}
    \lim_{t\to \infty}\Delta (z^t)^2 \leq (\Delta z^0)^2(1-r)\text{, where }r=\frac{\gamma\frac{\lambda_{n-2}}{\lambda_{1}}}{1-\frac{1-\gamma}{2n}\lambda_1}.
\end{equation}
\begin{theorem} \label{thm2}
Let $\Delta p_{b},\Delta p^{\infty}_{g}$ denote the precision error induced in a blockchain framework and the limiting asymptotic mean square error from BG respectively. $\Delta p^\infty_g\leq \Delta p_b$ if
\[ D \leq \frac{1}{\beta\sqrt{n(1-r)\Big[\Big(\frac{\Delta x^0}{x}\Big)^2 + \Big(\frac{\Delta y^0}{1+y}\Big)^2\Big]}}
\]
where $x^0,y^0$ represent the initial error of $x,y$ respectively in case of BG.
\end{theorem}
\begin{proof}
We consider the error on $p = x\cdot(1+y)$ to obtain
\begin{equation}
    \Delta p^2 = p^2 \Big[\Big(\frac{\Delta x}{x}\Big)^2 + \Big(\frac{\Delta y}{1+y}\Big)^2\Big].
\end{equation}
Analyzing error on $x,y$, we have
\begin{equation}
    \Delta x^2 = |x|^2\Big[\sum\limits_{i=1}^n\Big(\frac{\Delta x_i}{x_i}\Big)^2\Big], \text{ } \Delta y^2 = \sum\limits_{i=1}^n (\Delta y_i)^2.
\end{equation}
With a precision factor of $D$, we know that $\Delta x_i, \Delta y_i \leq \frac{1}{D}$ leading to the relations
\begin{equation}
\Big(\frac{\Delta x}{x}\Big)^2\leq \frac{1}{D^2}\sum\limits_{i=1}^n\Big(\frac{1}{x_i}\Big)^2, \text{ } \Big(\frac{\Delta y}{1+y}\Big)^2\leq \Big[\frac{n}{(1+y)D}\Big]^2.
\end{equation}
We know that,
\begin{equation}
  a_i = Pr(\sigma^i|s^i=0)Pr(s^i=0) \implies x_i = Pr(s^i=0|\sigma^i),  
\end{equation}
where,
\[
    x_i = \begin{cases}
        \beta, & \text{when } \sigma_i=1\\
        1-\alpha, & \text{when } \sigma_i=0
        \end{cases}
  \]
Let $p_b$ denote the global probability of no attack obtained with a blockchain based framework with precision factor $D$. Since $\beta,\alpha \approx 0$ and $ \frac{1}{x_i}\leq \frac{1}{\beta}$, we obtain
\begin{equation}
\Big(\frac{\Delta p_b}{p_b}\Big)^2 \leq  \frac{1}{D^2}\Big[\frac{n}{\beta^2}+ \Big(\frac{n}{1+y}\Big)^2\Big].
\end{equation}
Since $\frac{1}{(1+y)^2}\leq1$, we can safely assume that $\frac{1}{(1+y)^2}\in O(\frac{1}{n\beta^2})$ thereby leading to the following error estimate for the blockchain based algorithm
\begin{equation}\label{eq:blkerr}
    \Big(\frac{\Delta p_b}{p_b}\Big)^2 \leq  \Big(\frac{n}{D}\Big)^2\cdot\frac{1}{n\beta^2}.
\end{equation}
On the other hand for the BG algorithm recall that $x = e^u$, where $u=\sum\limits_{i=1}^nu_i$, $u_i=\log x_i$. Based on error analysis we obtain
\begin{equation}
    (\Delta x)^2 = \Big(\frac{\partial x}{\partial u} \Delta u\Big)^2 \implies \Big(\frac{\Delta x}{x}\Big)^2 = \Delta u^2.
\end{equation}
Let $p_g$ denote the global probability of no attack obtained from a BG based framework. We can therefore say that,
\begin{equation}
    \Big(\frac{\Delta p_g}{p_g}\Big)^2 \leq  \Delta u^2+ \Big(\frac{\Delta y}{1+y}\Big)^2
\end{equation}
Substituting the expected upper bound on error terms $\Delta u^2, \Delta y^2$ from Equation \eqref{eq:bgerr}, we obtain 
\begin{equation}\label{eq:gerr}
    \Big(\frac{\Delta p_g}{p_g}\Big)^2 \leq  n^2(1-r)\Big[(\Delta u^0)^2 + \Big(\frac{\Delta y^0}{1+y}\Big)^2\Big]
\end{equation}
Based on Equations \eqref{eq:blkerr} and \eqref{eq:gerr}, for the BG algorithm to completely outperform the blockchain based method, the following condition must be satisfied
\begin{equation}
    n^2(1-r)\Big[(\Delta u^0)^2 + \Big(\frac{\Delta y^0}{1+y}\Big)^2\Big] \leq \Big(\frac{\Delta p_b}{p_b}\Big)^2 \leq \Big(\frac{n}{D}\Big)^2\cdot\frac{1}{n\beta^2}.
\end{equation}
Since we know that $\Delta u^0 = \frac{x^0}{x}$ we can state that
\begin{equation}
    D \leq \frac{1}{\beta\sqrt{n(1-r)\Big[\Big(\frac{\Delta x^0}{x}\Big)^2 + \Big(\frac{\Delta y^0}{1+y}\Big)^2\Big]}}.
\end{equation}
\end{proof}
Theorem \ref{thm2} shows us that the precision factor $D$ has an inverse relation to the initial BG error, $\Delta x^0/x +\Delta y_0^2/y$. It can also be noted that $r(\gamma)$ is a monotonously increasing function with $0\leq r(\gamma) \leq \frac{\lambda_{n-2}}{\lambda_1}$ \cite{bg}. Therefore, Theorem \ref{thm2} indicates a direct relationship between the precision factor $D$ and the mixing parameter $\gamma$ as well as the eigenvalue ratio denoted by $\frac{\lambda_{n-2}}{\lambda_1}$. 

Based on observations from Theorem \ref{thm2}, we can postulate several constraints on BG in order for it to match the detection quality of a blockchain driven approach with a high precision factor $D$. First, BG must preferably start with low initial error with respect to the global values of $x,y$. Second, since $\gamma \to 1$ is more favorable for BG, only scant perturbation of the local estimate can be allowed when new neighbor messages are received. Therefore, a BG framework cannot afford a drastic change in the overall network mean, making a low initial error state on all nodes imperative to its success. Lastly, for good detection quality, the BG prefers a highly connected underlying graph as indicated by the requirement for $\frac{\lambda_{n-2}}{\lambda_1}\to1$.

%

The prerequisite constraints for the BG to outperform the blockchain present significant implementation challenges especially in a large scale power network with rapidly evolving global cyber health status. As a result, a blockchain driven framework is a highly favorable option for delivering an accurate, reliable and timely estimate of the global attack probability.  
\begin{figure}[t!]
\includegraphics[width=0.47\textwidth,keepaspectratio]{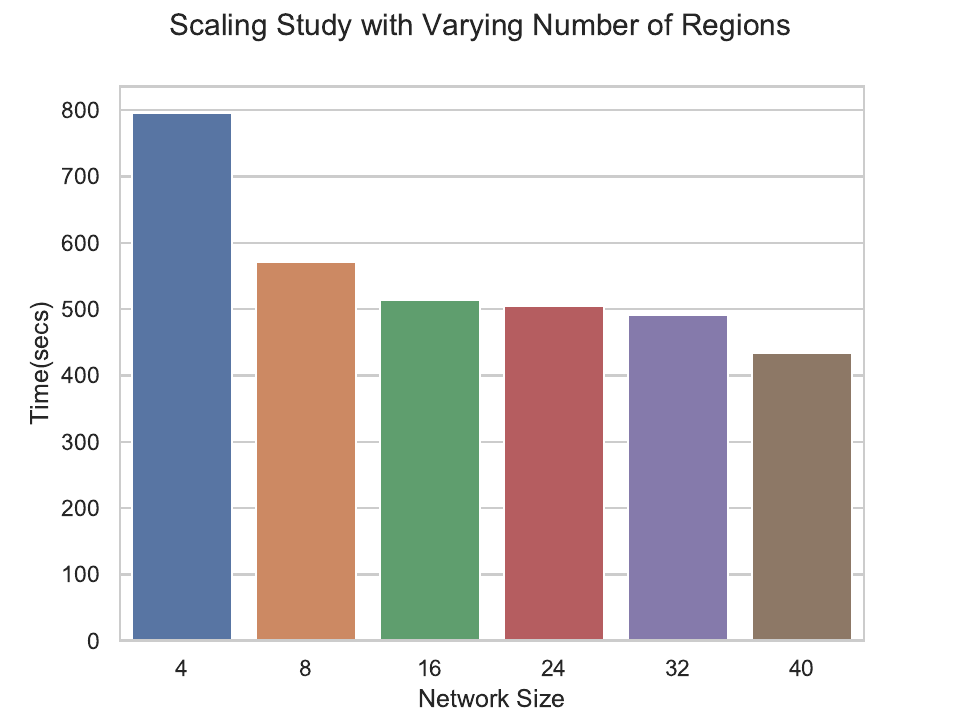}
\caption{Computational Scaling Study comparing Simulation Time against varying number of total regions for IEEE 3012 bus case}
\label{fig:res}
\end{figure}
\begin{figure*}[!htb]
\subfigure[At most 2 regions under attack]{\includegraphics[trim={0 0 0.5cm 0.8cm},clip,width=0.24\textwidth,keepaspectratio]{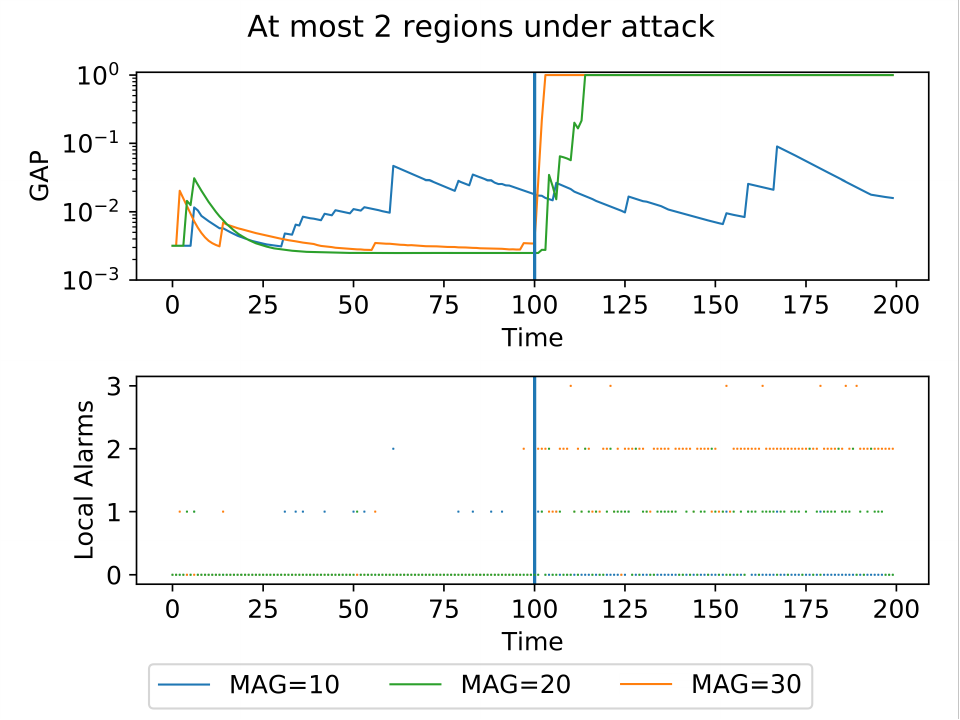}}\label{fig:2NUMA}
\subfigure[At most 4 regions under attack]{\includegraphics[trim={0 0 0.5cm 0.8cm},clip,width=0.24\textwidth,keepaspectratio]{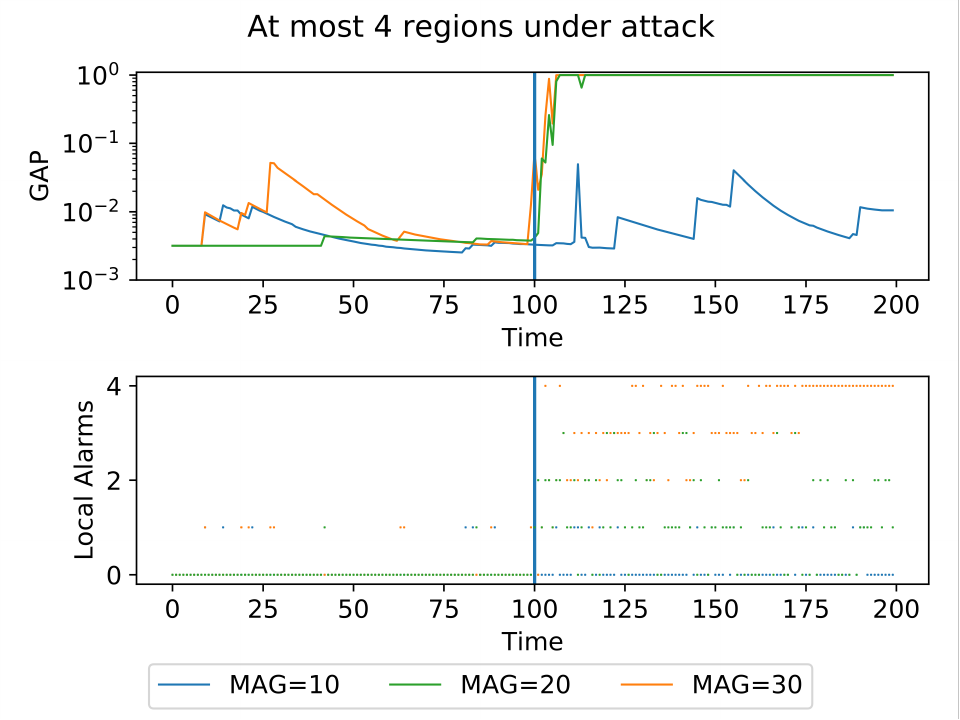}}\label{fig:4NUMA}
\subfigure[At most 8 regions under attack]{\includegraphics[trim={0 0 0.5cm 0.8cm},clip,width=0.24\textwidth,keepaspectratio]{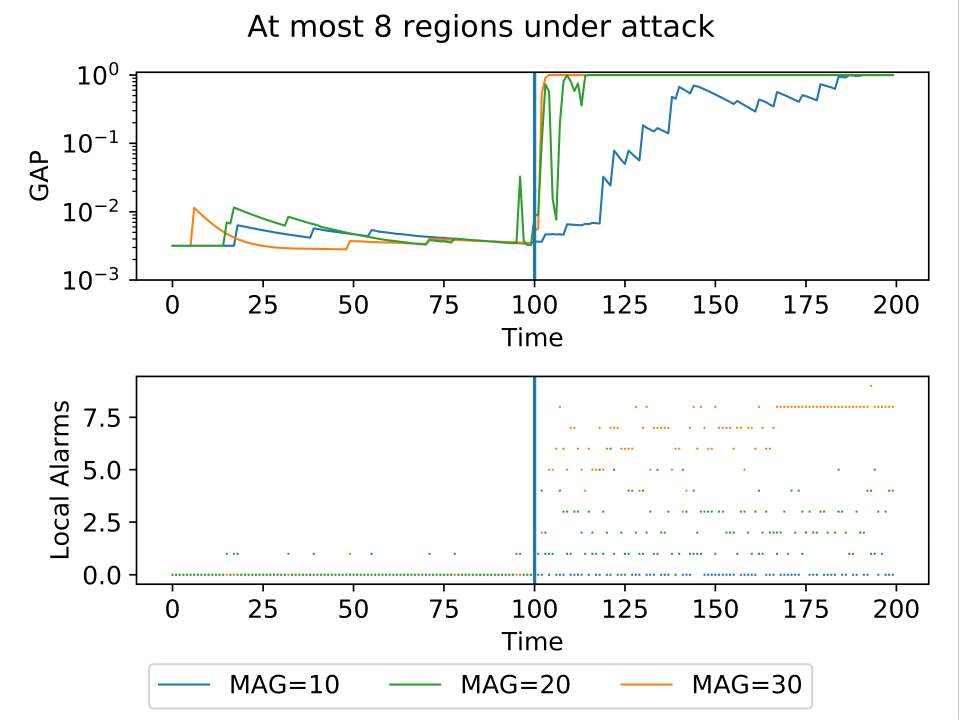}}\label{fig:8NUMA}
\subfigure[All 16 regions under attack]{\includegraphics[trim={0 0 0.5cm 0.8cm},clip,width=0.24\textwidth,keepaspectratio]{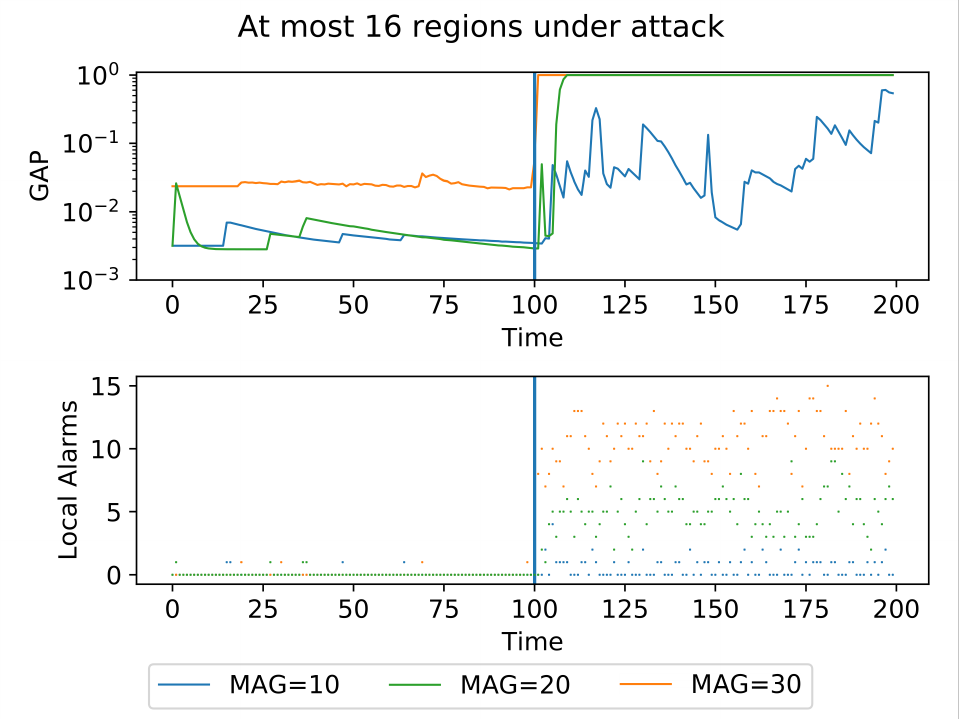}}\label{fig:16NUMA}
\caption{16 region decomposition of IEEE 3012 bus case : Global Attack Probability detection with varying number of regions under attack}
\label{fig:expres}
\end{figure*}
\begin{figure*}[!htb]
\subfigure[1 gossip round per epoch]{\includegraphics[trim={0 0 0cm 1cm},clip,width=0.24\textwidth,keepaspectratio]{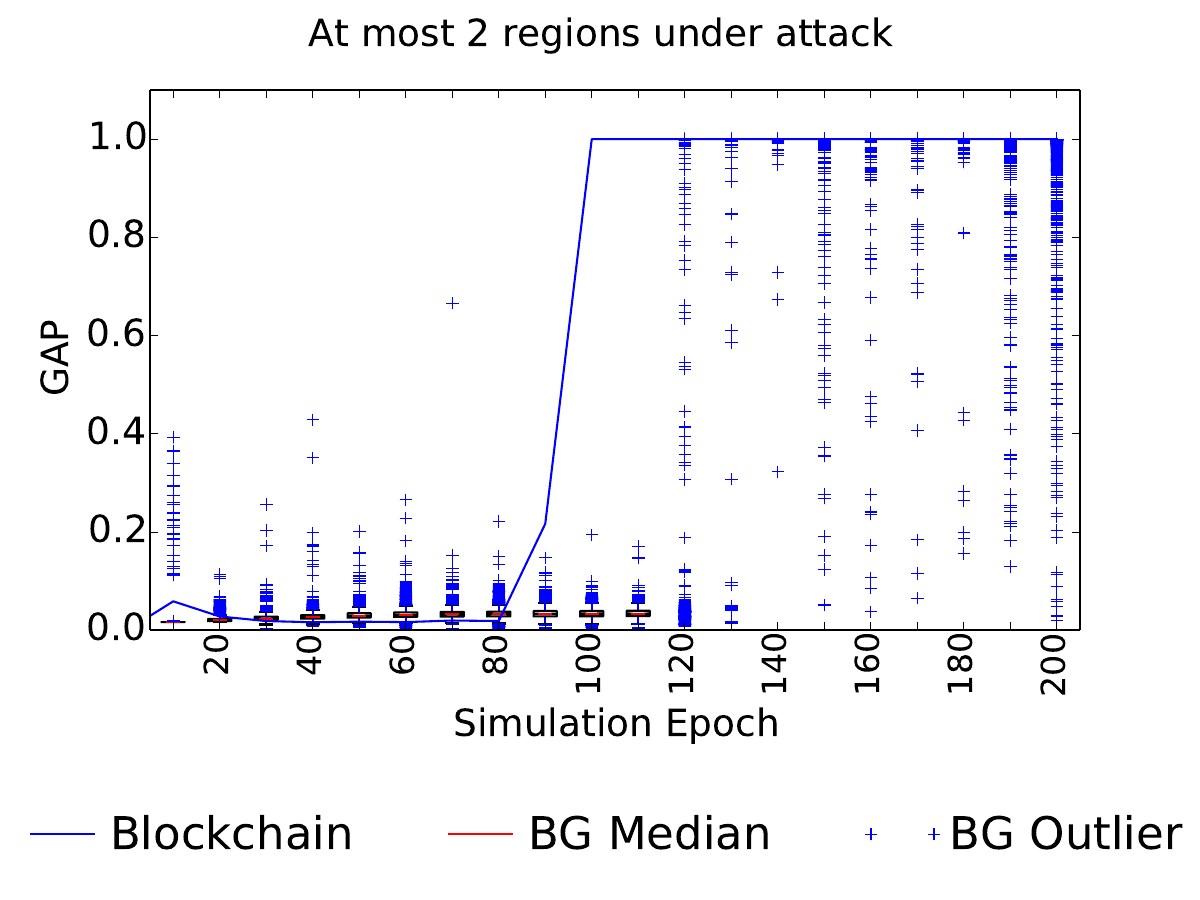}}\label{fig:1BGP}
\subfigure[10 gossip rounds per epoch]{\includegraphics[trim={0 0 0cm 1cm},clip,width=0.24\textwidth,keepaspectratio]{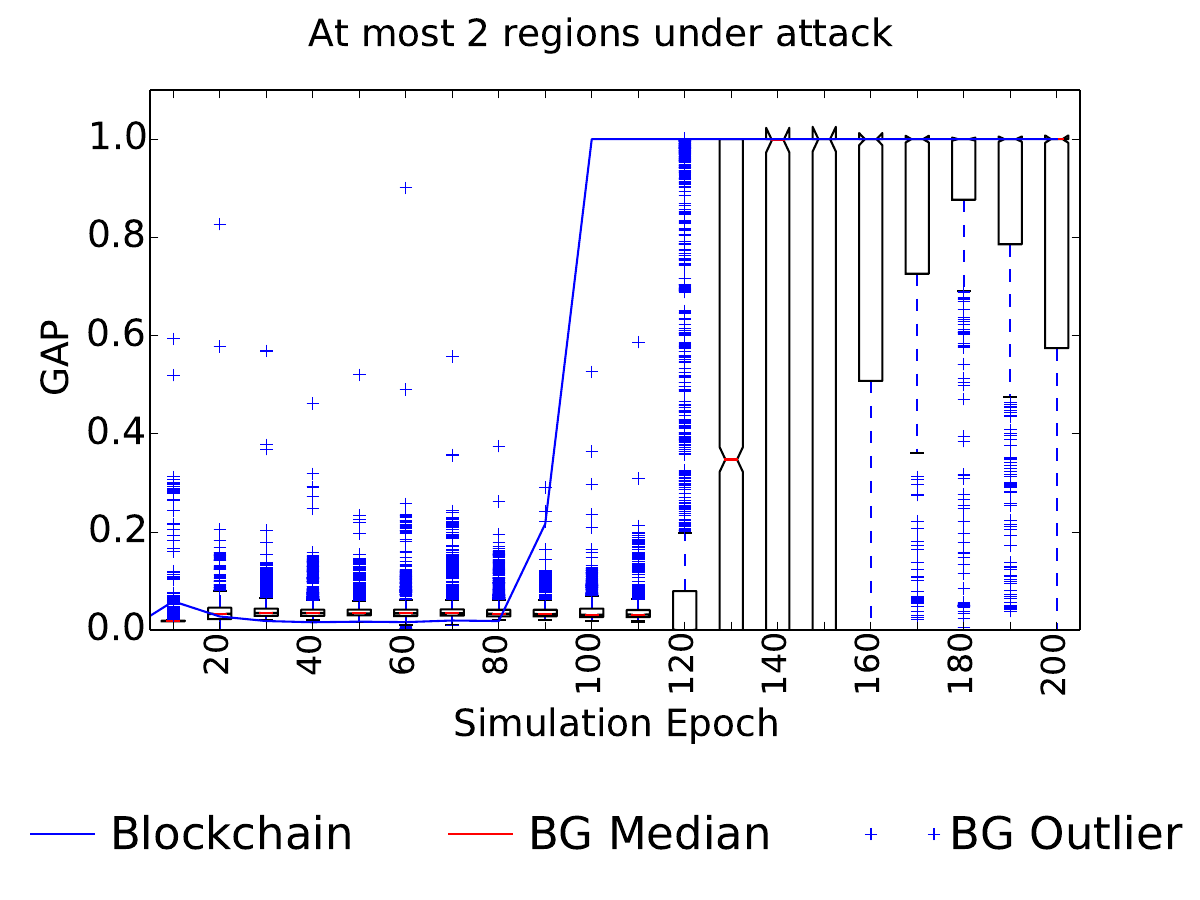}}\label{fig:10BGP}
\subfigure[25 gossip rounds per epoch]{\includegraphics[trim={0 0 0cm 1cm},clip,width=0.24\textwidth,keepaspectratio]{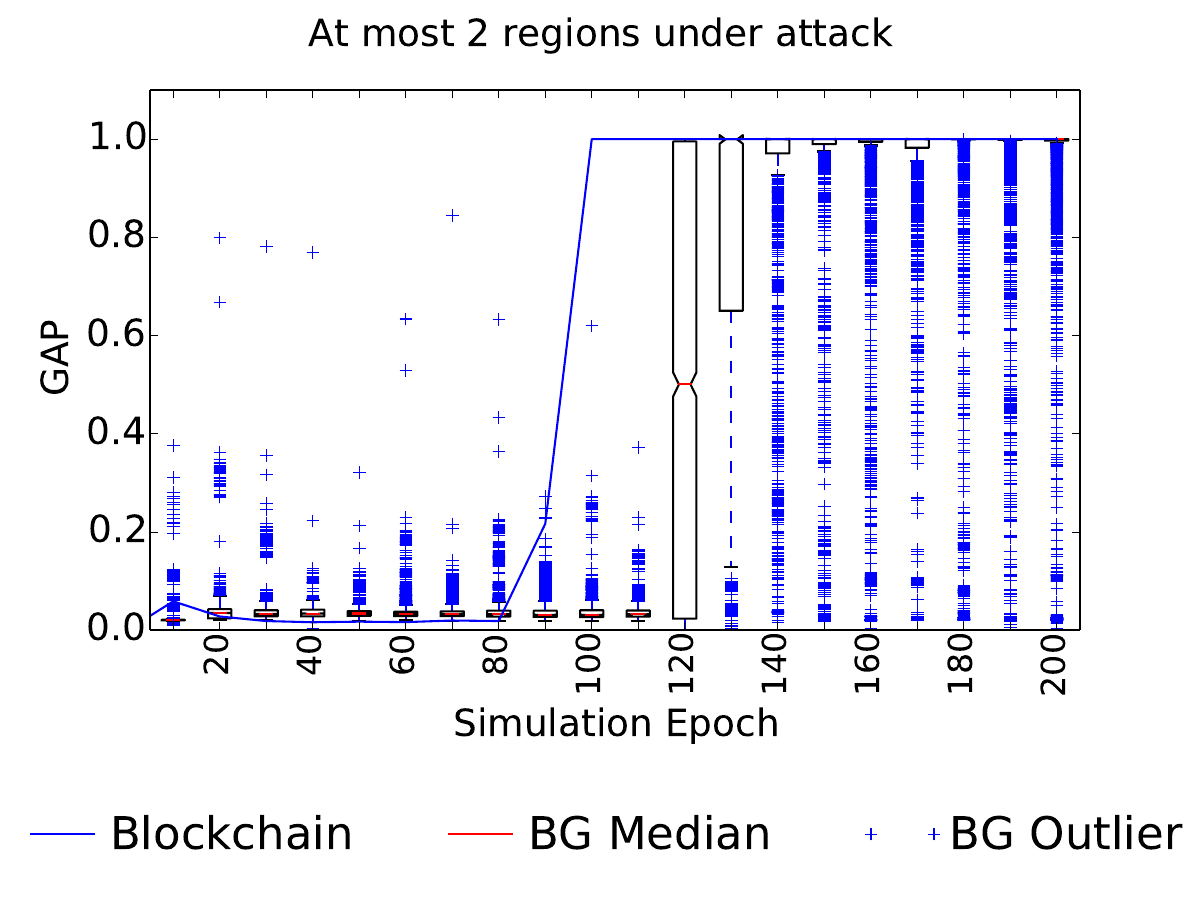}}\label{fig:25BGP}
\subfigure[50 gossip rounds per epoch]{\includegraphics[trim={0 0 0cm 1cm},clip,width=0.24\textwidth,keepaspectratio]{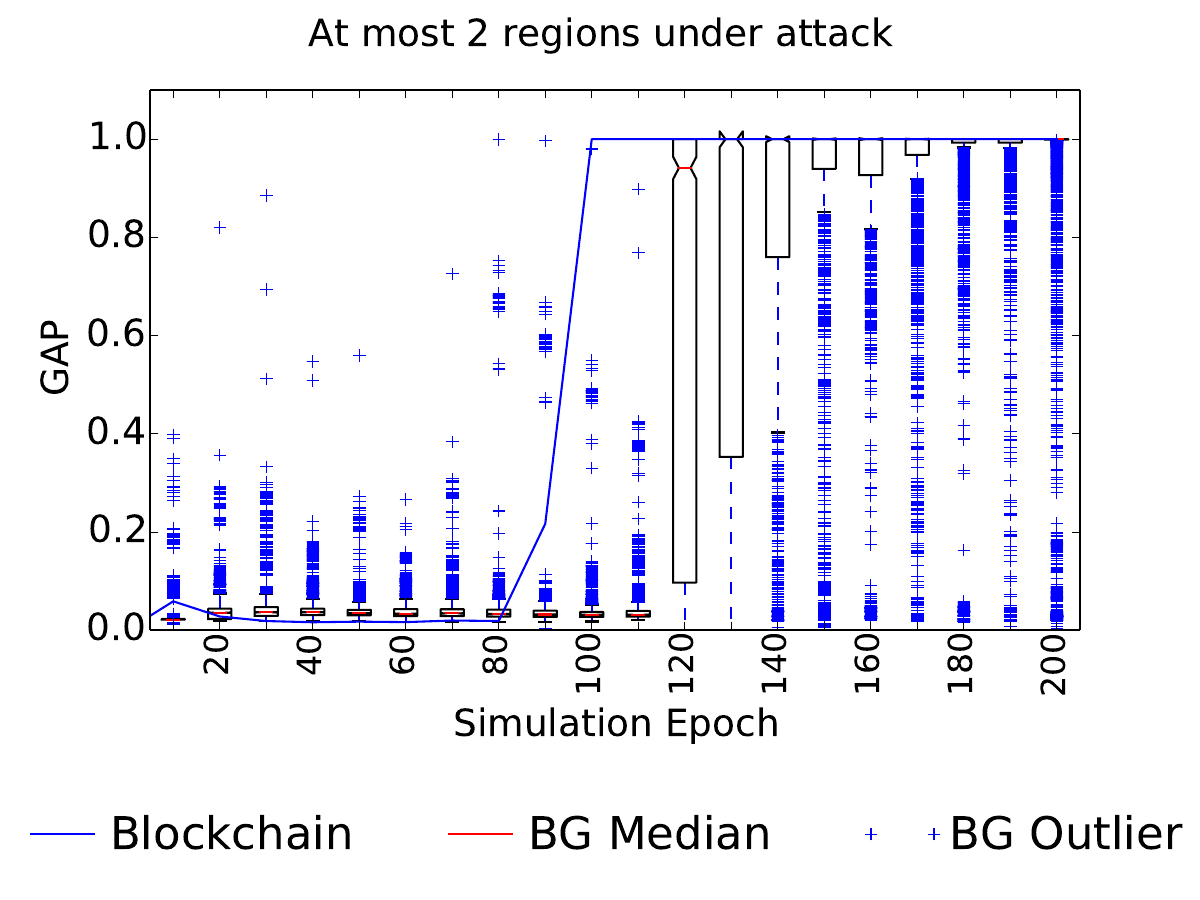}}\label{fig:50BGP}
\caption{40 region decomposition of IEEE 3012 bus case : Comparison with BG with varying local gossip rounds with fixed 2 regions under attack}
\label{fig:bglc}
\end{figure*}
\begin{figure*}[!htb]
\subfigure[At most 2 regions under attack]{\includegraphics[trim={0 0 0cm 1cm},clip,width=0.33\textwidth,keepaspectratio]{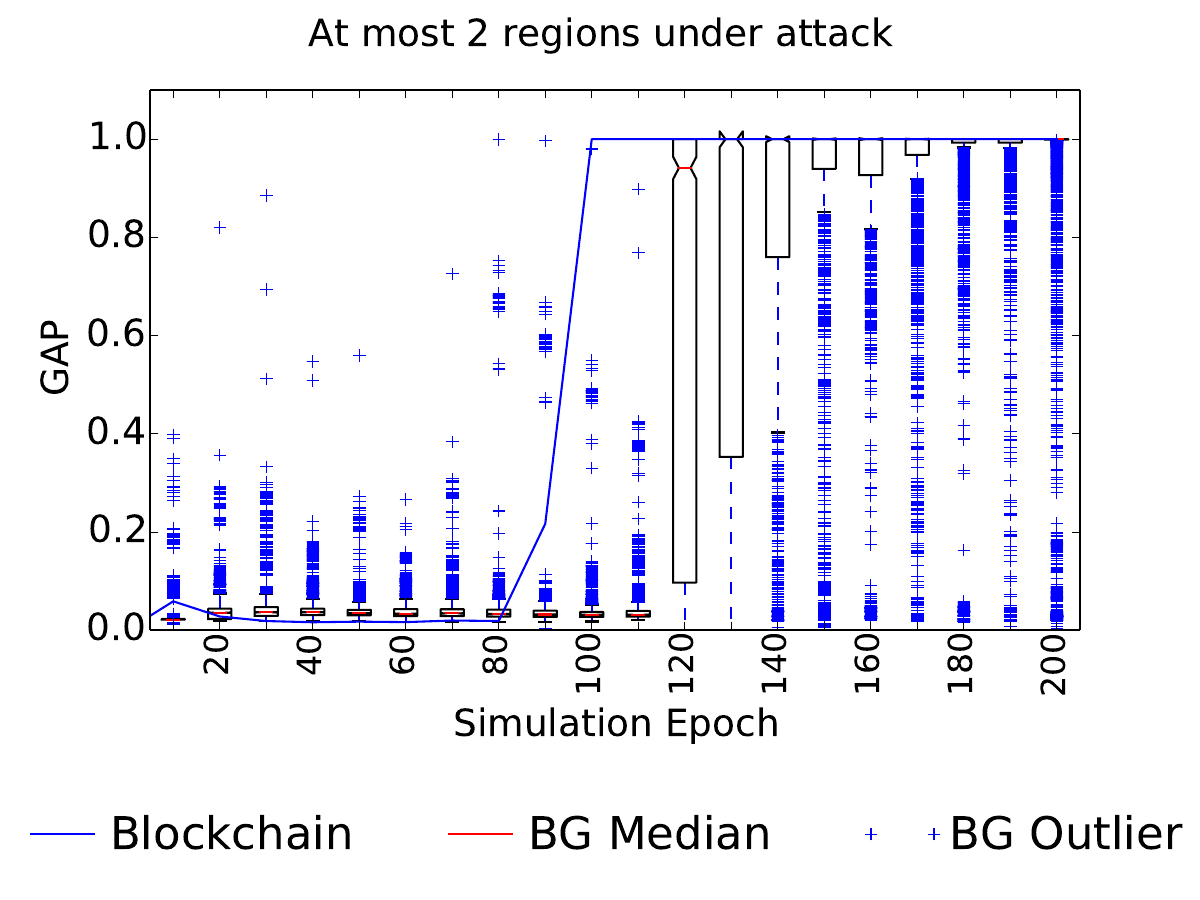}}\label{fig:2BGA}
\subfigure[At most 4 regions under attack]{\includegraphics[trim={0 0 0cm 1cm},clip,width=0.33\textwidth,keepaspectratio]{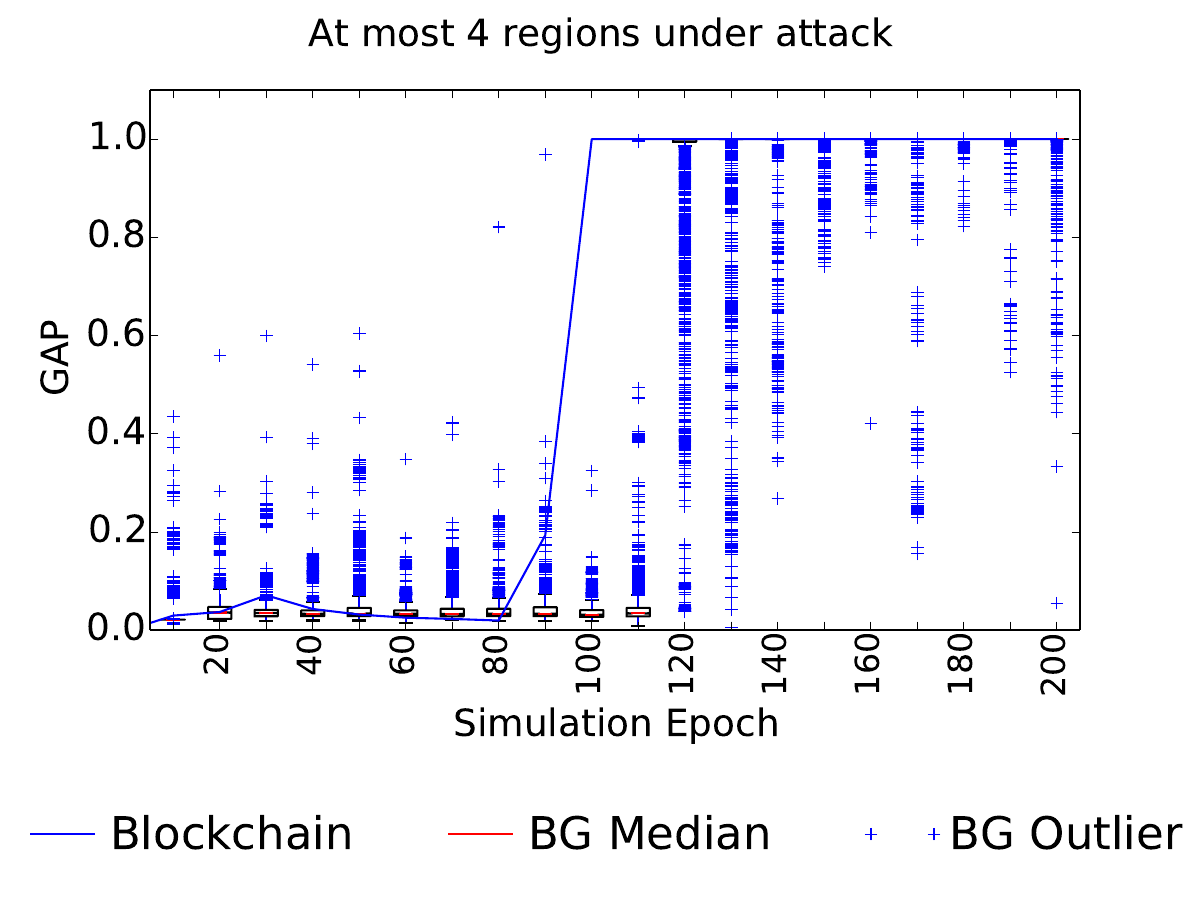}}\label{fig:4BGA}
\subfigure[At most 8 regions under attack]{\includegraphics[trim={0 0 0cm 1cm},clip,width=0.33\textwidth,keepaspectratio]{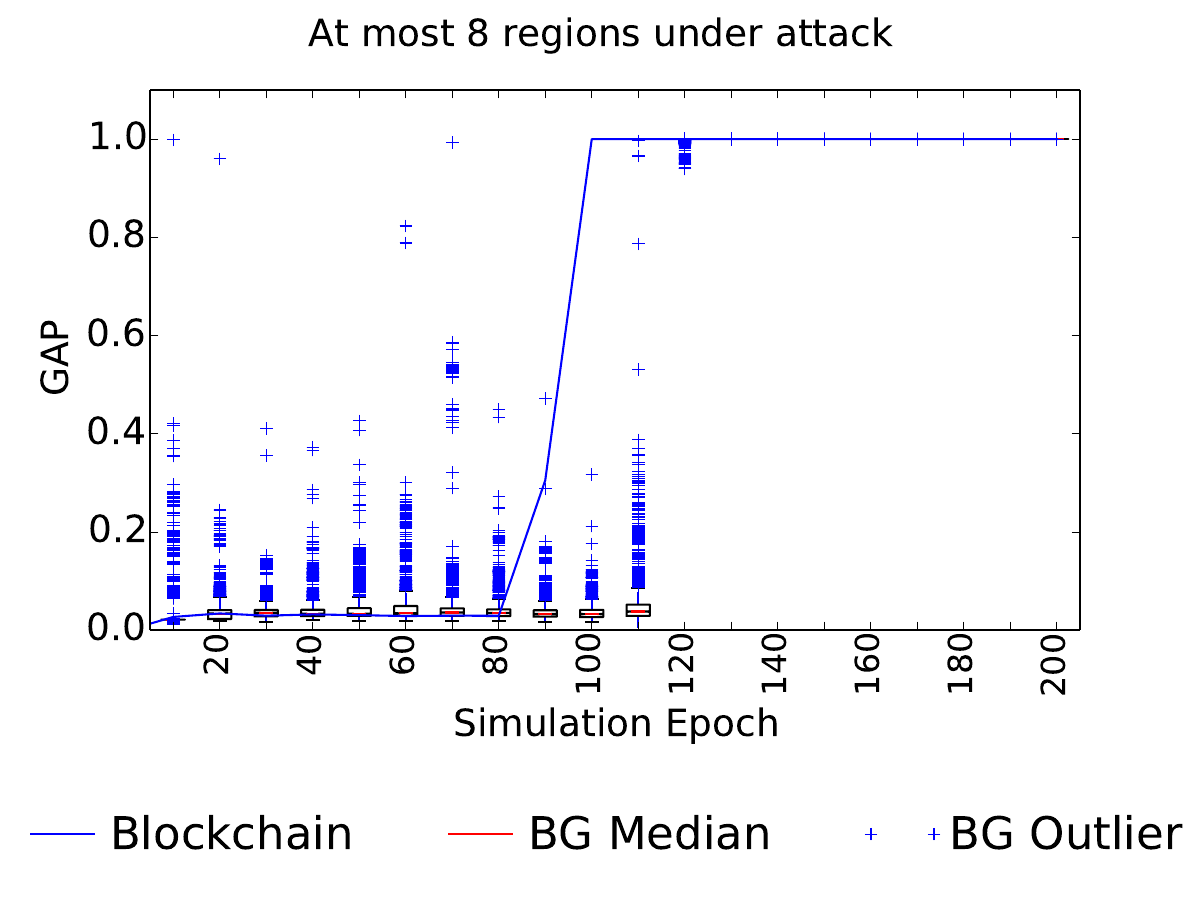}}\label{fig:8BGA}
\caption{40 region decomposition of IEEE 3012 bus case : Comparison with BG with varying number of regions under attack with local gossip rounds fixed at 50}
\label{fig:bgnuma}
\end{figure*}
\section{Experimental Results}
Our blockchain results utilize a Geth\footnote{https://geth.ethereum.org} based, Ethereum private blockchain network orchestrated on a 2.7 GHz Intel Xeon processor with 128 cores. We used \textit{Solidity 5.2} to compile the SC and utilized multithreading to simulate the region processes. Each region is assigned one thread each for the Ethereum node and the regional detection model respectively. There is no global synchronization among the regional processes. We used the IEEE 3012 bus transmission network 
to simulate a large scale transmission network consisting of 150 generators. In our experiments, we set $\alpha=0.005$ for the regional detection component for all regions. For our blockchain framework, we use a precision factor value of $D=10^6$. Further, we employ the Proof-of-Authority (PoA) consensus protocol \footnote{\url{https://wiki.parity.io/Proof-of-Authority-Chains}} given its suitability for power system applications wherein utilities can act as the authorities. PoA eliminates computationally intensive consensus by shifting the onus of trust to the authorities themselves. Our simulations occur at discrete epochs at which local alarms are reported. In all our experiments the simulation horizon is 200 epochs.

{ 
We leverage the Clique Proof-of-Authority consensus protocol which offers numerous benefits from the security perspective. First, Clique relies on the authority of real world entities that are a part of a permissioned ledger. Based on Table \ref{tab:poa_comparison}, a majority (i.e. 50\%) of the authorities have to reach consensus before a block is permanently added to the chain. Moreover, Clique can also tolerate an adversarial attack of at most 50\% of the validators/authorities. As a result, PoA improves security without relying on computational challenges, since an attacker must hack a majority of authorities in order to rescind all transactions. PoA chains are also known to have low latency, deterministic block creation process as well as faster block creation times which are more important in a permissioned setting. PoA is more suited to consortium or permissioned settings due to its security, easier maintenance as well as accountability of the authorities themselves. Therefore, in this paper, we adopt a permissioned blockchain driven approach since our problem setting involves a consortium of real world entities such as utility stakeholders with their own ICSs desirous of detecting global replay attacks.} 
\subsection{Computational Performance and Accuracy}\label{sec:performance}
We evaluate the computational performance of our blockchain framework by considering varied regional decompositions of the IEEE 3012 bus case.
Figure \ref{fig:res} presents the incurred time for simulation with varying number of total regions for our blockchain based framework. We observe that a lesser number of regions takes longer to simulate in general. We also notice that the overall simulation time decreases sharply after 4 regions, decreasing gradually and eventually settling at approximately 500secs. This is because as a result of the Clique \footnote{\url{ https://github.com/ethereum/EIPs/issues/225}} PoA consensus protocol, lower values of $n$ contribute to a high overall simulation time due to increased block generation time \cite{POA_analysis}. However, increasing $n$, results in lesser block generation time but more communication burden. Therefore, we observe that the simulation time eventually stabilizing around a fixed value even as $n$ is increased.

Figure \ref{fig:expres} presents results pertaining to experiments involving varied number of regions under attack along with varying degree of attack magnitude. We observe that the framework is robust to varying attack magnitudes and is able to qualitatively detect attacks regardless of attack magnitude. Further, the results also indicate that the framework performs consistently with varying number of regions under attack as well.
\subsection{Evaluation against BG}
Figures \ref{fig:bglc} and \ref{fig:bgnuma} depict box and line plots against discrete simulation epochs corresponding to the BG and the blockchain respectively for a magnitude of 20. Each box plot presents the spectrum of global probability values recorded by all region across 100 repeated trials of BG. Similarly, the line plot represents the global probability values determined by the blockchain framework.

Figure \ref{fig:bglc} compares the performance of the blockchain version with a BG based scheme with varying number of gossip rounds for each simulation epoch by enforcing two regional attacks. We observe that the blockchain based framework detects a global attack faster than the BG by a margin of 30 epochs. Moreover an increasing number of BG rounds per epoch leads to a sharp decrease in the mean, median and the variance of the global attack probability across all regions. Figure \ref{fig:bglc} demonstrates that despite 50 synchronous gossip rounds per epoch, a global attack detection cannot be successfully disseminated among all regions in BG.

Figure \ref{fig:bgnuma} compares the performance of the blockchain based global attack detection paradigm with its BG counterpart with increasing number of regions under attack. We fix the maximum number of gossip rounds to 50 for each simulation case. We observe that the blockchain outperforms the gossip algorithm by a margin of approximately 20 simulation epochs. Figure \ref{fig:bgnuma} also indicates that in case of BG the initial disparity of the GAP among the regions decreases with an increase in the number of regions under attack. Such a scenario results in  marginal improvement in the global attack detection. 

As postulated by Theorem \ref{thm2}, low initial error is preferred for BG, which only happens when an increasing number of regions are under attack. On the other hand, a high number of gossip rounds are required for BG in order to overcome regional connectivity constraints predicted by Theorem \ref{thm2}. Both scenarios highlight the operational obstacles associated with an information diffusion scheme based on gossip. Figures \ref{fig:bglc}, \ref{fig:bgnuma} conclusively show that the blockchain based framework delivers a reliable, timely and accurate detection as compared to its gossip counterpart in diverse operational scenarios.

\section{Conclusion and Future Work}
In this paper, we present a decentralized blockchain based global attack detection mechanism that only uses locally reported alarm and its associated statistics to detect the onset of a global replay attack. We design a novel Bayesian update mechanism requiring one global multiplication and one global addition leading to a scalable and computationally efficient blockchain paradigm. We characterize the performance of the blockchain based global attack detection mechanism against a broadcast gossip based counterpart. In order to do so, we first reformulate the computation of the global multiplication and addition operations to be amenable towards broadcast gossip. We then theoretically analyze the performance of the broadcast gossip with a limited precision blockchain version. Our analysis predicts an overall computational superiority of the blockchain version as opposed to the broadcast gossip. We implement and evaluate the blockchain based approach on a private Ethereum network with the help of Solidity for orchestrating the Smart Contracts. Our experiments demonstrate the accuracy of the decentralized detection mechanism as well as its robustness to increasing number of regions. Moreover, our results also corroborate our theoretical claim of computational superiority over the state-of-the-art, decentralized broadcast gossip paradigm by a significant margin. For our future work we plan to focus on local attacks that could be mutually correlated. We also aim to investigate a blockchain driven approach to distinguish between multiple failure modes and an attack.

\bibliographystyle{IEEEtran}   
\bibliography{main}

{ 
\section*{Appendix A}\label{sec:appendixA}
We provide a description of all the characteristic features listed in Table \ref{tab:poa_comparison} for permissioned consensus mechanisms. The features and their descriptions are as follows:
\begin{itemize}[leftmargin=*]
\item \textit{Validators} refers to the minimum number of entities required for consensus to take place. 
\item \textit{Decentralization} aspect characterizes the potential of a few entities dominating over block creation. 
\item \textit{Fault tolerance} describes the maximum number of entities that must fail to jeopardize the blockchain operation. 
\item \textit{Consistency} refers to the final state of the distributed ledger among multiple entities post consensus. 
\item \textit{Block latency} characterizes the degree of communication required in order to mint a block. 
\end{itemize}

Decentralization potential of Clique and Aura are deemed high due to the fact that all participating entities are provided an opportunity to propose blocks during run time \cite{POA_analysis}. In IBFT, validators once chosen need to be added and removed explicitly by a proposal followed by voting \cite{saltini2019ibft}. In Raft, if the leader is malicious, then the entire blockchain can be compromised, unless a change is effected using voting of the verifiers \cite{salimitari2018survey}. Therefore, the decentralization potential for IBFT and Raft is deemed medium and low respectively.

With respect to consistency, Clique relies on the Ethereum GHOST protocol to eventual resolve forks and achieve consistency of the chain. However, Aura makes no such guarantees \cite{POA_analysis}. On the other hand, IBFT emphasizes consistency over speed resulting in a stalling in cases where a fork has emerged \cite{POA_analysis}. Raft meanwhile assumes that no forking can emerge on account of there being only one leader \cite{salimitari2018survey}}
\end{document}